\documentclass[a4paper,10pt,twocolumn]{article}
\usepackage[footnotesize]{caption}
\usepackage[caption=false,font=footnotesize]{subfig}
\usepackage{graphicx,amsmath,amsfonts,amsthm,amssymb,url,cite,bm,bbm}
\usepackage[margin={0.575in,0.8in}]{geometry}

\begin{document}

\newcommand{\eps}{\varepsilon}
\newcommand{\var}{\operatorname{Var}}
\newcommand{\expn}{\mathbb{E}}
\newcommand{\pr}{\mathbb{P}}
\newcommand{\pmark}{\mathcal{P}_{\text{mark}}}
\renewcommand{\vec}{\bm}

\newcommand{\thint}{\vec{\theta}_{\text{int}}}
\newcommand{\thall}{\vec{\theta}_{\text{all1}}}
\newcommand{\thmaj}{\vec{\theta}_{\text{maj}}}
\newcommand{\thmin}{\vec{\theta}_{\text{min}}}
\newcommand{\thcoi}{\vec{\theta}_{\text{coin}}}
\newcommand{\thadd}{\vec{\theta}_{\text{add}}}
\newcommand{\thdil}{\vec{\theta}_{\text{dil}}}
\newcommand{\ththr}{\vec{\theta}_{\text{thr}}}
\newcommand{\ththrr}{\vec{\theta}_{\text{thrr}}}

\newcommand{\tos}{^s}
\newcommand{\caps}{C\tos}
\newcommand{\pints}{p_{\text{int}}\tos}
\newcommand{\palls}{p_{\text{all1}}\tos}
\newcommand{\pmajs}{p_{\text{maj}}\tos}
\newcommand{\pmins}{p_{\text{min}}\tos}
\newcommand{\pcois}{p_{\text{coin}}\tos}
\newcommand{\padds}{p_{\text{add}}\tos}
\newcommand{\pdils}{p_{\text{dil}}\tos}
\newcommand{\pthrs}{p_{\text{thr}}\tos}
\newcommand{\pthrrs}{p_{\text{thrr}}\tos}

\newcommand{\toj}{^j}
\newcommand{\capj}{C\toj}
\newcommand{\pintj}{p_{\text{int}}\toj}
\newcommand{\pallj}{p_{\text{all1}}\toj}
\newcommand{\pmajj}{p_{\text{maj}}\toj}
\newcommand{\pminj}{p_{\text{min}}\toj}
\newcommand{\pcoij}{p_{\text{coin}}\toj}
\newcommand{\paddj}{p_{\text{add}}\toj}
\newcommand{\pdilj}{p_{\text{dil}}\toj}
\newcommand{\pthrj}{p_{\text{thr}}\toj}
\newcommand{\pthrrj}{p_{\text{thrr}}\toj}

\newtheorem{theorem}{Theorem}
\newtheorem{lemma}{Lemma}
\newtheorem{proposition}{Proposition}
\newtheorem{corollary}{Corollary}
\newtheorem{conjecture}{Conjecture}

\title{Capacities and Capacity-Achieving Decoders\\ for Various Fingerprinting Games}
\author{Thijs Laarhoven\footnote{T. Laarhoven is with the Department of Mathematics and Computer Science, Eindhoven University of Technology, P.O. Box 513, 5600 MB Eindhoven, The Netherlands. \protect\\
E-mail: mail@thijs.com.}}
\date{\today}

\maketitle
\begin{abstract}
Combining an information-theoretic approach to fingerprinting with a more constructive, statistical approach, we derive new results on the fingerprinting capacities for various informed settings, as well as new log-likelihood decoders with provable code lengths that asymptotically match these capacities. The simple decoder built against the interleaving attack is further shown to achieve the simple capacity for unknown attacks, and is argued to be an improved version of the recently proposed decoder of Oosterwijk et al. With this new universal decoder, cut-offs on the bias distribution function can finally be dismissed.

Besides the application of these results to fingerprinting, a direct consequence of our results to group testing is that (i) a simple decoder asymptotically requires a factor $1.44$ more tests to find defectives than a joint decoder, and (ii) the simple decoder presented in this paper provably achieves this bound.
\end{abstract}


\section{Introduction}

To protect copyrighted content against unauthorized redistribution, distributors may embed watermarks or fingerprints in the content, uniquely linking copies to individual users. Then, if an illegal copy of the content is found, the distributor can extract the watermark from the copy and compare it to the database of watermarks, to determine which user was responsible. 

To combat this solution, pirates may try to form a coalition of several colluders, each owning a differently watermarked copy of the content, and perform a collusion attack. By comparing their different versions of the content, they will detect differences in their copies which must be part of the watermark. They can then create a new pirate copy, where the resulting watermark matches the watermark of different pirates in different segments of the content, making it hard for the distributor to find the responsible users. Fortunately, under the assumption that if the pirates don't detect any differences (because they all received the same version) they output this watermark (known in the literature as the Boneh-Shaw marking assumption~\cite{boneh98}), it is still possible to find all colluders using suitable fingerprinting codes.


\subsection{Model}

The above fingerprinting game is often modeled as the following two-person game between the distributor $\mathcal{D}$ and the coalition of pirates $\mathcal{C}$. The set of colluders is assumed to be a random subset of size $|\mathcal{C}| = c$ from the complete set of $n$ users $\mathcal{U}$, and the identities of these colluders are unknown to the distributor. The aim of the game for the distributor is ultimately to discover the identities of the colluders, while the colluders want to stay hidden. The game consists of the following three phases: (i) the distributor uses an \textit{encoder} to generate the fingerprints; (ii) the colluders employ a \textit{collusion channel} to generate the pirate output, and (iii) the distributor uses a \textit{decoder} to map the pirate output to a set of accused users.

\paragraph{Encoder}
First, the distributor generates a fingerprinting code $\mathcal{X}$, consisting of $n$ code words $\vec{X}_1, \dots, \vec{X}_n$ from $\{0,1\}^{\ell}$.\footnote{More generally $\mathcal{X}$ is a code with code words of length $\ell$ from an alphabet $\mathcal{Q}$ of size $q \geq 2$, but in this paper we restrict our attention to the binary case $q = 2$.} The $i$th entry of code word $j$ indicates which version of the content is assigned to user $j$ in the $i$th segment. The parameter $\ell$ is referred to as the code length, and the distributor would like $\ell$ to be as small as possible.

A common restriction on the encoding process is to assume that $\mathcal{X}$ is created by first generating a probability vector $\vec{P} \in [0,1]^{\ell}$ by choosing each entry $P_i$ independently from a certain distribution function $F$, and then generating $\mathcal{X}$ according to $\pr(X_{j,i} = 1) = P_i$. This guarantees that watermarks of different users $j$ are independent, and that watermarks in different positions $i$ are independent. Schemes that satisfy this assumption are sometimes called \textit{bias-based schemes}, and the encoders discussed in this paper also belong to this category.

\paragraph{Collusion channel}
After generating $\mathcal{X}$, the entries are used to select and embed watermarks in the content, and the content is sent out to all users. The colluders then get together, compare their copies, and use a certain collusion channel or pirate attack $\vec{\theta}$ to select the pirate output $\vec{Y} \in \{0,1\}^{\ell}$. If the pirate attack behaves symmetrically both in the colluders and in the positions $i$, then the collusion channel can be modeled by a vector $\vec{\theta} \in [0,1]^{c+1}$, consisting of entries $\theta_z = f_{Y|Z}(1|z)$ indicating the probability of outputting a $1$ when the pirates received $z$ ones and $c - z$ zeroes. Some common attacks $\vec{\theta}$ are described in Section~\ref{sec:attacks}.

\paragraph{Decoder}
Finally, after the pirate output has been generated and distributed, we assume that the distributor intercepts it and applies a decoding algorithm to the pirate output $\vec{Y}$, the code $\mathcal{X}$ and the (secret) bias vector $\vec{P}$ to compute a set $\mathcal{C}' \subseteq \mathcal{U}$ of accused users. This is commonly done by assigning \textit{scores} to users, and accusing those users whose score exceeds some predefined threshold $\eta$. The distributor wins the game if $\mathcal{C}'$ is non-empty and contains only colluders (i.e.\ $\emptyset \neq \mathcal{C}' \subseteq \mathcal{C}$) and loses if this is not the case, which could be because an innocent user $j \notin \mathcal{C}$ is falsely accused (a false positive error), or because no guilty users are accused (a false negative error). We often write $\eps_1$ and $\eps_2$ for upper bounds on the false positive and false negative probabilities respectively.


\subsection{Related work}

Work on the above bias-based fingerprinting game started in 2003, when Tardos proved that any fingerprinting scheme must satisfy $\ell \propto c^2 \ln n$, and that a bias-based scheme is able to achieve this optimal scaling in $\ell$~\cite{tardos03}. He proved the latter by providing a simple and explicit construction with a code length of $\ell = 100 c^2 \ln(n/\eps_1)$, which is known in the literature as the Tardos scheme.

\paragraph{Improved constructions}
Later work on the constructive side of fingerprinting focused on improving upon Tardos' result by sharpening the bounds~\cite{blayer08, skoric08b}, optimizing the distribution functions~\cite{nuida07}, improving the score function~\cite{skoric08}, tightening the bounds again with this improved score function~\cite{laarhoven13ihmmsec, laarhoven14dcc, nuida09, simone12, skoric08, skoric13}, optimizing the score function~\cite{oosterwijk13}, and again tightening the bounds with this optimized score function~\cite{ibrahimi13, oosterwijk13b} to finally end up with a sufficient asymptotic code length of $\ell \sim 2 c^2 \ln n$ for large $n$. This construction can be extended to larger alphabets, in which case the code length scales as $\ell \sim 2 c^2 \ln(n)/(q - 1)$. Other work on practical constructions focused on joint decoders, which are computationally more involved but may work with shorter codes~\cite{meerwald12, moulin08, oosterwijk14}, and side-informed fingerprinting games~\cite{charpentier09, furon09b, laarhoven13wifs, oosterwijk13}, where estimating the collusion channel $\vec{\theta}$ was considered to get an improved performance. 

Recently Abbe and Zheng~\cite{abbe10} showed that, in the context of fingerprinting~\cite{meerwald12}, if the set of allowed collusion channels satisfies a certain one-sidedness condition, then a decoder that achieves capacity against the information-theoretic worst-case attack is a universal decoder achieving capacity against arbitrary attacks. The main drawback of using this result is that the worst-case attack is hard to compute, but this does lead to more insight why e.g.\ Oosterwijk et al.~\cite{oosterwijk13b} obtained a universal decoder by considering the decoder against the `interleaving attack', which is known to be the asymptotic worst-case attack. 

\paragraph{Fingerprinting capacities}
At the same time, work was also done on establishing bounds on the fingerprinting capacity $C$, which translate to lower bounds on the required asymptotic code length $\ell$ through $\ell \gtrsim C^{-1} \log_2 n$ for large $n$. For the binary case Huang and Moulin~\cite{huang09, huang09b, huang10, huang12, moulin08} and Amiri and Tardos~\cite{amiri09} independently derived exact asymptotics for the fingerprinting capacity for arbitrary attacks as $C \sim (2 c^2 \ln 2)^{-1}$, corresponding to a minimum code length of $\ell \sim 2 c^2 \ln n$. Huang and Moulin~\cite{huang12} further showed that to achieve this bound, an encoder should use the arcsine distribution $F^*$ for generating biases $p$: 
\begin{align}
F^*(p) = \frac{2}{\pi} \arcsin \sqrt{p}. \qquad \quad (0 < p < 1)
\end{align}
These capacity-results were later generalized to the $q$-ary setting~\cite{boesten11, huang12b} showing that a $q$-ary code length of $\ell \sim 2 c^2 \ln(n)/(q-1)$ is asymptotically optimal.

\paragraph{Dynamic fingerprinting}
There has also been some interest in a variant of the above fingerprinting game where several rounds of the two-player game between the distributor and the coalition are played sequentially. This allows the distributor to adjust the encoding and decoding steps of the next rounds to the knowledge obtained from previous rounds. Many of the bias-based constructions can also be used effectively in this dynamic setting~\cite{laarhoven12wifs, laarhoven13tit, laarhoven13wifs} with equivalent asymptotics for the required code length, but allowing the distributor to trace all colluders even if the collusion channel is not symmetric in the colluders, and leading to significantly smaller first order terms than in the `static' setting. These bias-based dynamic schemes may even be able to compete with the celebrated scheme of Fiat and Tassa~\cite{fiat01}.

\paragraph{Group testing}
Finally, a different area of research closely related to fingerprinting is that of group testing, where the set of $n$ users corresponds to a set of $n$ items, the set of $c$ colluders corresponds to a subset of $c$ defective items, and where the aim of the distributor is to find all defective items by performing group tests. This game corresponds to a special case of the fingerprinting game, where the pirate attack is fixed in advance (and possibly known to the distributor) to (a variant of) the `all-$1$ attack'. In this game it is significantly easier to find all pirates/defectives; it is known that a joint decoder asymptotically requires only $\ell \sim c \log_2 n$ tests~\cite{sebo85}, while simple decoders exist requiring as few as $\ell \sim e c \ln n$ tests to find all defectives~\cite{chan12}. Recent work has shown that applying results from fingerprinting to group testing may lead to improved results compared to what is known in the group testing literature~\cite{laarhoven13allerton, meerwald11b}.


\subsection{Contributions}

\sloppypar{In this work we first extend the work of Huang and Moulin~\cite{huang12} by deriving explicit asymptotics for the simple and joint capacities of various fingerprinting games with different amounts of side-information. Table~\ref{tab:1} summarizes tight lower bounds on the code length constant for various informed settings obtained via the capacities. These asymptotics can be seen as our `targets' for the second part of this paper, which describes decoders with provable bounds on $\ell$ and $\eta$ that asymptotically achieve these capacities. In fact, if the collusion channel that the decoder was built against matches the attack used by the pirates, then the proof that the resulting simple decoders achieve capacity is remarkably simple and holds for arbitrary attacks.}

\begin{table}[!t]
\renewcommand{\arraystretch}{1.3}
\caption{Asymptotics for tight lower bounds on $L = \ell/\ln n \cong C^{-1}/\ln 2$, based on the simple and joint capacities, with different amounts of side information (see Section~\ref{sec:sideinformation}). The proposed simple decoders are shown to match these bounds, and we conjecture that the proposed joint decoders are also asymptotically optimal.\label{tab:1}}
\centering
\begin{tabular}{p{2.5cm}cccc} \hline
 & \multicolumn{2}{c}{Fully informed} & \multicolumn{2}{c}{Partially informed} \\ 
 & Simple & Joint & Simple & Joint \\ \hline
Interleaving atk. & $2c^2$ & $2c^2$ & $2c^2$ & $2c^2$ \\ 
All-$1$ attack & $2.08 c$ & $1.44 c$ & $1.83 c\sqrt{c}$ & $1.32 c\sqrt{c}$ \\
Majority voting & $3.14 c$ & $1.44 c$ & $2.41 c\sqrt{c}$ & $1.20 c\sqrt{c}$ \\ 
Minority voting & $2.08 c$ & $1.44 c$ & $0.66 c\sqrt{c}$ & $0.43 c\sqrt{c}$ \\ 
Coin-flip attack & $8.33 c$ & $4.48 c$ & $5.18 c\sqrt{c}$ & $2.32 c\sqrt{c}$ \\ \hline
\end{tabular}
\end{table}

\paragraph{Capacity-achieving simple decoding without cut-offs}
Similar to Oosterwijk et al.~\cite{oosterwijk13, oosterwijk13b}, who studied the decoder built against the interleaving attack because that attack is in a sense optimal, we then turn our attention to the simple decoder designed against the interleaving attack, and argue that it is an improved version of Oosterwijk et al.'s universal decoder. To provide a sneak preview of this result, the new score function is the following:
\begin{align}
g(x,y,p) = \begin{cases} 
\ln\left(1 + \frac{p}{c(1 - p)}\right) & x = y = 0 \\ 
\ln\left(1 - \frac{1}{c}\right) & x \neq y \\ 
\ln\left(1 + \frac{1 - p}{c p}\right) & x = y = 1 
\end{cases}
\end{align}
This decoder is shown to achieve the uninformed simple capacity, and we argue that with this decoder (i) the Gaussian assumption always holds (and convergence to the normal distribution is much faster), and (ii) no cut-offs on the bias distribution function $F$ are ever needed anymore.

\paragraph{Joint log-likelihood decoders}
Since it is not hard to extend the definition of the simple decoder to joint decoding, we also present and analyze joint log-likelihood decoders. Analyzing these joint decoders turns out to be somewhat harder due to the `mixed tuples', but we give some motivation why these decoders seem to work well. We also conjecture that the joint decoder tailored against the interleaving attack achieves the joint uninformed capacity, but proving this result is left for future work. 

\paragraph{Applications to group testing}
Since the all-$1$ attack in fingerprinting is equivalent to a problem known in the literature as group testing~\cite{laarhoven13wifs, meerwald11b}, some of our results can also be applied to this area. In fact, we derive two new results in the area of group testing: (i) any simple-decoder group testing algorithm requires at least $\ell \sim \log_2(e)^2 c \ln n \approx 2.08 c \ln n$ group tests to find $c$ defective items hidden among $n$ items, and (ii) the decoder discussed in Section~\ref{sec:dec-simple} provably achieves this optimal scaling in $\ell$. This decoder was previously considered in~\cite{meerwald11b}, but no provable bounds on the (asymptotic) code lengths were given there.


\subsection{Outline}

The outline of the paper is as follows. Section~\ref{sec:models} first describes the various different models we consider in this paper, and provides a roadmap for Sections~\ref{sec:cap} and \ref{sec:dec}. Section~\ref{sec:cap} discusses capacity results for each of these models, while Section~\ref{sec:dec} discusses decoders which aim to match the lower bounds on $\ell$ obtained in Section~\ref{sec:cap}. Finally, in Section~\ref{sec:discussion} we conclude with a brief discussion of the most important results and remaining open problems.


\section{Different models}
\label{sec:models}

Let us first describe how the results in Sections~\ref{sec:cap} and \ref{sec:dec} are structured according to different assumptions, leading to different models. Besides the general assumptions on the model discussed in the introduction, we further make a distinction between models based on (1) the computational complexity of the decoder, (2) the information about $\vec{\theta}$ known to the distributor, and (3) the collusion channel used by the pirates. These are discussed in Sections~\ref{sec:complexity}, \ref{sec:sideinformation} and \ref{sec:attacks} respectively.


\subsection{Decoding complexity}
\label{sec:complexity}

Commonly two types of decoders are considered, which use different amounts of information to decide whether a user should be accused or not. 
\begin{enumerate}
  \item \sloppypar{\textbf{Simple decoding}: To quote Moulin~\cite[Section 4.3]{moulin08}: \textit{``The receiver makes an innocent/guilty decision on each user independently of the other users, and there lies the simplicity but also the suboptimality of this decoder.''} In other words, the decision to accuse user $j$ depends only on the $j$th code word of $\mathcal{X}$, and not on other code words from $\mathcal{X}$.}
  \item \sloppypar{\textbf{Joint decoding}: In this case, the decoder is allowed to base the decision whether to accuse a user on the entire code $\mathcal{X}$. Such decoders may be able to obtain smaller code lengths than possible with the best simple decoders.}
\end{enumerate}
Using more information generally causes the time complexity of the decoding step to go up, so usually there is a trade-off between a shorter code length and a faster decoding algorithm.


\subsection{Side-informed distributors}
\label{sec:sideinformation}

We consider three different scenarios with respect to the knowledge of the distributor about the collusion channel $\vec{\theta}$. Depending on the application, different scenarios may apply.
\begin{enumerate}
  \item \textbf{Fully informed}: Even before $\mathcal{X}$ is generated, the distributor already knows exactly what the pirate attack $\vec{\theta}$ will be. This information can thus be used to optimize both the encoding and decoding phases. This scenario applies to various group testing models, and may apply to dynamic traitor tracing, where after several rounds the distributor may have estimated the pirate strategy.
  \item \textbf{Partially informed}: The tracer does not know in advance what collusion channel will be used, so the encoding is aimed at arbitrary attacks. However, after obtaining the pirate output $\vec{y}$, the distributor does learn more about $\vec{\theta}$ before running an accusation algorithm, e.g.\ by estimating the attack based on the available data. So the encoding is uninformed, but we assume that the decoder is informed and knows $\vec{\theta}$. Since the asymptotically optimal bias distribution function $F$ in fingerprinting is known to be the arcsine distribution $F^*$, we will assume that $F^*$ is used for generating biases. This scenario is similar to EM decoding~\cite{charpentier09, furon09b}.
  \item \textbf{Uninformed}: In this case, both the encoding and decoding phases are assumed to be done without prior knowledge about $\vec{\theta}$, so also the decoder should be designed to work against arbitrary attacks. This is the most commonly studied fingerprinting game.
\end{enumerate}
For simplicity of the analysis, in the partially informed setting we assume that the estimation of the collusion channel is precise, so that $\vec{\theta}$ is known exactly to the decoder. This assumption may not be realistic, but at least we can then obtain explicit expressions for the capacities, and get an idea of how much estimating the strategy may help in reducing the code length. This also allows us to derive explicit lower bounds on $\ell$: even if somehow the attack can be estimated correctly, then the corresponding capacities tell us that we will still need at least a certain number of symbols to find the pirates.


\subsection{Common collusion channels}
\label{sec:attacks}

As mentioned in the introduction, we assume that collusion channels satisfy the marking assumption, which means that $\theta_0 = 0$ and $\theta_c = 1$. For the remaining values of $z \in \{1, \dots, c-1\}$ the pirates are free to choose how often they want to output a $1$ when they receive $z$ ones. Some commonly considered attacks are listed below. 
\begin{enumerate}
  \item \textbf{Interleaving attack}: The coalition randomly selects one of its members and outputs his symbol. This corresponds to $(\thint)_z = z/c$. This attack is known to be asymptotically optimal (from the point of view of the colluders) in the uninformed max-min fingerprinting game~\cite{huang12}.
  \item \textbf{All-$1$ attack}: The pirates output a $1$ whenever they can, i.e., whenever they have at least one $1$. This translates to $(\thall)_z = \mathbbm{1}\{z > 0\}$. This attack is of particular interest due to its relation with group testing.
  \item \textbf{Majority voting}: The colluders output the most common symbol among their received symbols. This means that $(\thmaj)_z = \mathbbm{1}\{z > c/2\}$.
  \item \textbf{Minority voting}: The traitors output the symbol which they received the least often (but received at least once). For $1 \leq z \leq c - 1$, this corresponds to $(\thmin)_z = \mathbbm{1}\{z < c/2\}$.
  \item \textbf{Coin-flip attack}: If the pirates receive both symbols, they flip a fair coin to decide which symbol to output. So for $1 \leq z \leq c - 1$, this corresponds to $(\thcoi)_z = \frac{1}{2}$.
\end{enumerate}
For even $c$, defining $\theta_{c/2}$ in a consistent way for majority and minority voting is not straightforward. For simplicity, in the analysis of these two attacks we will therefore assume that $c$ is odd. Note that in the uninformed setting, we do not distinguish between different collusion channels; the encoder and decoder should then work against arbitrary attacks.


\subsection{Roadmap}

The upcoming two sections about capacities (Section~\ref{sec:cap}) and decoders (Section~\ref{sec:dec}) are structured according to the above classification, where first the decoding complexity is chosen, then the side-information is fixed, and finally different attacks are considered. For instance, to find the joint capacity in the fully informed game one has to go to Section~\ref{sec:cap-joint-informed}, while the new simple uninformed decoder can be found in Section~\ref{sec:dec-simple-uninformed}.


\section{Capacities}
\label{sec:cap}

\sloppypar{In this section we establish lower bounds on the code length of any valid decoder, by inspecting the information-theoretic capacities of the various fingerprinting games. We will use some common definitions from information theory, such as the binary entropy function $h(x) = -x \log_2 x - (1 - x) \log_2(1-x)$, the relative entropy or Kullback-Leibler divergence $d(x\|y) = x \log_2(x/y) + (1-x) \log_2((1-x)/(1-y))$, and the mutual information $I(X;Y) = \sum_{x,y} \pr(x,y) \log_2 (\pr(x,y) / \pr(x) \pr(y))$. The results in this section build further upon previous work on this topic by Huang and Moulin~\cite{huang12}.}


\subsection{Simple capacities}
\label{sec:cap-simple}

For simple decoders, we assume that the decision whether to accuse user $j$ is based solely on $\vec{X}_j$, $\vec{Y}$ and $\vec{P}$. Focusing on a single position, and denoting the random variables corresponding to a colluder's symbol, the pirate output, and the bias in this position by $X_1$, $Y$ and $P$, the interesting quantity to look at~\cite{huang09} is the mutual information $I(X_1;Y|P = p)$. This quantity depends on the pirate strategy $\vec{\theta}$ and on the bias $p$. To study this mutual information we will use the following equality~\cite[Equation~(61)]{huang12},
\begin{align}
I(X_1;Y|P = p) = p d(a_1 \| a) + (1 - p) d(a_0 \| a),
\end{align}
where $a, a_0, a_1$ are defined as
\begin{align}
a &= \sum_{z = 0}^c \binom{c}{z} p^z (1 - p)^{c-z} \theta_z, \\
a_0 &= \sum_{z = 0}^{c-1} \binom{c - 1}{z} p^z (1 - p)^{c-z-1} \theta_z, \\
a_1 &= \sum_{z = 1}^c \binom{c - 1}{z - 1} p^{z-1} (1 - p)^{c-z} \theta_z.
\end{align}
Note that given $p$ and $\vec{\theta}$, the above formulas allow us to compute the associated mutual information explicitly. 


\subsubsection{Fully informed}
\label{sec:cap-simple-informed}

In the fully informed setting we are free to choose $F$ to maximize the capacity, given a collusion channel $\vec{\theta}$. When the attack is known to the distributor in advance, there is no reason to use different values of $p$; the distributor should always use the value of $p$ that maximizes the mutual information payoff $I(X_1; Y|P = p)$. Given an attack strategy $\vec{\theta}$, the capacity we are interested in is thus
\begin{align}
\caps(\vec{\theta}) = \max_p I(X_1; Y|P = p).
\end{align}
For general attacks finding the optimal value of $p$ analytically can be hard, but for certain specific attacks we can investigate the resulting expressions individually to find the optimal values of $p$ that maximize the mutual information. This leads to the following results for the five attacks listed in Section~\ref{sec:attacks}. Proofs will appear in the full version.

\begin{theorem} \label{thm:cap-simple-informed}
The simple informed capacities and the corresponding optimal values of $p$ for the five attacks of Section~\ref{sec:attacks} are:
\begin{alignat}{2}
\caps(\thint) &\sim \frac{1}{2c^2 \ln 2}, \qquad & \pints &= \frac{1}{2}, \tag{S1} \label{eqs1} \\ 
\caps(\thall) &\sim \frac{\ln 2}{c}, & \palls &\sim \frac{\ln 2}{c}, \tag{S2} \label{eqs2} \\%
\caps(\thmaj) &\sim \frac{1}{\pi c \ln 2}, & \pmajs &= \frac{1}{2}, \tag{S3} \label{eqs3} \\%
\caps(\thmin) &\sim \frac{\ln 2}{c}, & \pmins &\sim \frac{\ln 2}{c}, \tag{S4} \label{eqs4} \\
\caps(\thcoi) &\sim \frac{\ln 2}{4 c}, & \pcois &\sim \frac{\ln 2}{2c}. \tag{S5}\label{eqs5} 
\end{alignat}
\end{theorem}

Since fully informed protection against the all-$1$ attack is equivalent to noiseless group testing, and since the code length $\ell$ scales in terms of the capacity $C$ as $\ell \geq C^{-1} \log_2 n$, we immediately get the following corollary.

\begin{corollary} \label{cor:cap-simple-informed-group}
Any simple group testing algorithm for $c$ defectives and $n$ total items requires an asymptotic number of group tests $\ell$ of at least 
\begin{align}
\ell \sim \frac{c \ln n}{\ln(2)^2} \approx 2.08 c \ln n.
\end{align}
\end{corollary}

Note that this seems to contradict earlier results of \cite{laarhoven13allerton}, which suggested that under a certain Gaussian assumption, only $\ell \sim 2 c \ln n$ tests are required. This apparent contradiction is caused by the fact that the Gaussian assumption in~\cite{laarhoven13allerton} is not correct in the regime of small $p$, for which those results were derived. In fact, the distributions considered in that paper roughly behave like binomial distributions over $\ell$ trials with probability of success of $O(1/\ell)$, which converge to Poisson distributions. Numerical inspection shows that the relevant distribution tails are indeed not very Gaussian and do not decay fast enough. Rigorous analysis of the scores in~\cite{laarhoven13allerton} shows that an asymptotic code length of about $3 c \ln n$ is sufficient when $p \sim \ln(2)/c$, which is well above the lower bound of Corollary~\ref{cor:cap-simple-informed-group}. Details can be found in the full version.


\subsubsection{Partially informed}
\label{sec:cap-simple-part}

If the encoder is uninformed, then the best he can do against arbitrary attacks (for large $c$) is to generate biases using the arcsine distribution $F^*$. So instead of computing the mutual information in one point $P = p$, we now average over different values of $p$ where $p$ follows the arcsine distribution. So the capacity we are interested in is given by
\begin{align}
\caps(\vec{\theta}) = \expn_p I(X_1;Y|P = p) = \int_0^1 \frac{I(X_1;Y|P = p)}{\pi \sqrt{p(1 - p)}} \, dp.
\end{align}
The resulting integrals are hard to evaluate analytically, even for large $c$, although for some collusion channels we can use Pinsker's inequality (similar to the proof of~\cite[Theorem 7]{huang12}) to show that $\caps(\vec{\theta}) = \Omega(c^{-3/2})$. And indeed, if we look at the numerics of $c^{3/2} \caps(\vec{\theta})$ in Figure~\ref{fig:cap-simple-part}, it seems that the partially informed capacity usually scales as $c^{-3/2}$. As a consequence, even if the attack can be estimated exactly, then still a code length of the order $\ell \propto c^{3/2} \ln n$ is required to get a scheme that works. Note that for the interleaving attack, the capacity scales as $c^{-2}$.

\begin{figure}
\centering
\includegraphics[width=\columnwidth]{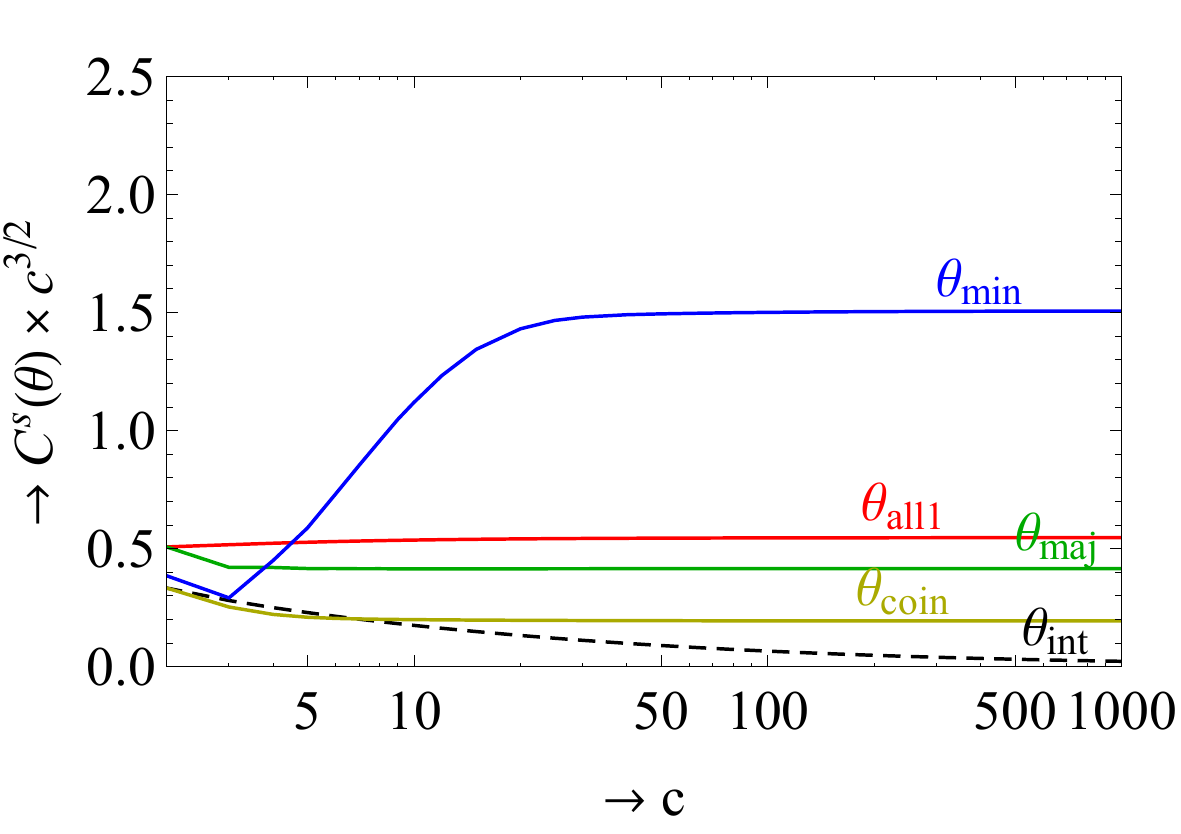}
\caption{The partially informed simple capacity (multiplied by $c^{3/2}$) as a function of $c$ for different pirate attacks. From top to bottom these curves correspond to minority voting, the all-$1$ attack, majority voting, the coin-flip attack, and the interleaving attack. Except for the interleaving attack, for which the capacity scales as $c^{-2}$ (the dashed line), these capacities all seem to scale as $c^{-3/2}$. \label{fig:cap-simple-part}}
\end{figure}

%
%
%


\subsubsection{Uninformed}
\label{sec:cap-simple-uninformed}

For the uninformed fingerprinting game, where both the encoder and decoder are built to work against arbitrary attacks, we are interested in the following max-min game:
\begin{align}
\caps = \max_F \min_{\vec{\theta}(F)} \expn_p I(X_1;Y|P = p).
\end{align}
Huang and Moulin~\cite{huang12, huang12b} previously solved this uninformed game for asymptotically large coalition sizes $c$ as follows.

\begin{proposition} \cite[Theorem 3]{huang12b}
The simple uninformed capacity is given by
\begin{align}
\caps \sim \frac{1}{2 c^2 \ln 2},
\end{align}
and the optimizing encoder $F$ and collusion channel $\vec{\theta}$ achieving this bound for large $c$ are the arcsine distribution $F^*$ and the interleaving attack $\thint$.
\end{proposition}

Note that while for the interleaving attack the capacity is the same (up to order terms) for each of the three side-informed cases, for the four other attacks the capacity gradually increases from $O(c^{-2})$ to $O(c^{-3/2})$ to $O(c^{-1})$ when the distributor is more and more informed.


\subsection{Joint capacities}
\label{sec:cap-joint}

If the computational complexity of the decoder is not an issue, joint decoding may be an option. In that case, the relevant quantity to examine is the mutual information between the symbols of all colluders, denoted by $X_1, \dots, X_c$, and the pirate output $Y$, given $P$: $I(X_1, \dots, X_c; Y|P = p)$~\cite{huang12}. Note that $Y$ only depends on $X_1, \dots, X_c$ through $Z = \sum X_i$, so $I(X_1, \dots, X_c; Y|P = p) = I(Z; Y|P = p)$. To compute the joint capacities, we use the following convenient explicit formula~\cite[Equation~(59)]{huang12}:
\begin{align}
\frac{1}{c} I(Z; Y|P = p) = \frac{1}{c} \left[h(a) - a_h\right],
\end{align}
where $h(p) = -p \log_2 p - (1 - p) \log_2(1 - p)$ is the binary entropy function, and $a_h$ is defined as 
\begin{align}
a_h &= \sum_{z = 0}^c \binom{c}{z} p^z (1 - p)^{c-z} h(\theta_z).
\end{align}


\subsubsection{Fully informed}
\label{sec:cap-joint-informed}

In the fully informed setting, the capacity is again obtained by considering the mutual information and maximizing it as a function of $p$:
\begin{align}
\capj(\vec{\theta}) = \frac{1}{c} \max_p I(Z;Y|P = p).
\end{align}
Computing this is very easy for the all-$1$ attack, the majority voting attack and the minority voting attack, since one can easily prove that the joint capacity is equal to $\frac{1}{c}$ whenever the collusion channel is deterministic, e.g. when $\theta_z \in \{0,1\}$ for all $z$. Since the capacity for the interleaving attack was already known, the only non-trivial case is the coin-flip attack. A proof of the following theorem can be found in the full version.

\begin{theorem}
The joint informed capacities and the corresponding optimal values of $p$ for the five attacks of Section~\ref{sec:attacks} are:
\begin{alignat}{2}
\capj(\thint) &\sim \frac{1}{2c^2 \ln 2}, & \pintj &= \frac{1}{2}, \tag{J1} \label{eqj1} \\
\capj(\thall) &= \frac{1}{c}, & \pallj &\sim \frac{\ln 2}{c}, \tag{J2} \label{eqj2} \\
\capj(\thmaj) &= \frac{1}{c}, & \pmajj &= \frac{1}{2}, \tag{J3} \label{eqj3} \\
\capj(\thmin) &= \frac{1}{c}, & \pminj &= \frac{1}{2}, \tag{J4} \label{eqj4} \\
\capj(\thcoi) &\sim \frac{\log_2(5/4)}{c}, \qquad & \pcoij &\sim \frac{\ln(5/3)}{c}. \tag{J5}\label{eqj5} 
\end{alignat}
\end{theorem}

Recall that there is a one-to-one correspondence between the all-$1$ attack and group testing, so the result above establishes firm bounds on the asymptotic number of group tests required by any probabilistic group testing algorithm. This result was already known, and was first derived by Seb\H{o}~\cite[Theorem 2]{sebo85}.


\subsubsection{Partially informed}
\label{sec:cap-joint-part}

For the partially informed capacity we again average over the mutual information where $p$ is drawn at random from the arcsine distribution $F^*$. Thus the capacity is given by
\begin{align}
\capj(\vec{\theta}) = \frac{1}{c} \expn_p I(Z;Y|P = p).
\end{align}
Exact results are again hard to obtain, but we can at least compute the capacities numerically to see how they behave. Figure~\ref{fig:cap-joint-part} shows the capacities of the five attacks of Section~\ref{sec:attacks}. Although the capacities are higher for joint decoding than for simple decoding, the joint capacities of all attacks but the interleaving attack also scale as $c^{-3/2}$.

\begin{figure}[!t]
\centering
\includegraphics[width=\columnwidth]{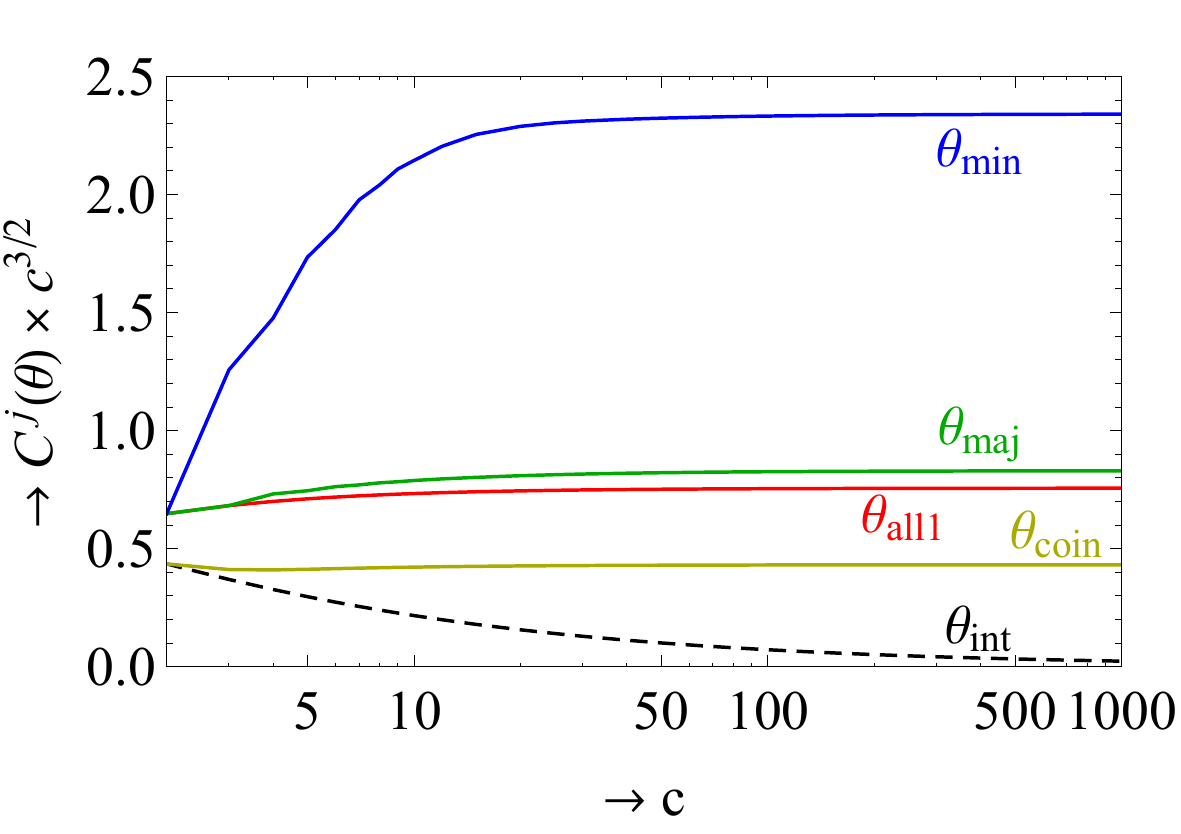}
\caption{The partially informed joint capacity, as a function of $c$, for different attacks. From top to bottom these are minority voting, majority voting, the all-$1$ attack, the coin-flip attack and the interleaving attack. Except for the interleaving attack, these capacities all seem to scale as $c^{-3/2}$. \label{fig:cap-joint-part}}
\end{figure}


\subsubsection{Uninformed}
\label{sec:cap-joint-uninformed}

Finally, if we are working with joint decoders which are supposed to work against arbitrary attacks, then we are interested in the following max-min mutual information game:
\begin{align}
\capj = \max_F \min_{\vec{\theta}(F)} \expn_p \frac{1}{c} I(Z;Y|P = p).
\end{align}
This joint capacity game was previously solved by Huang and Moulin~\cite{huang12} who showed that also in the joint game, the interleaving attack $\thint$ and the arcsine distribution $F^{*}$ together form a saddle-point solution to the uninformed fingerprinting game.

\begin{proposition} \cite[Theorem 6, Corollary 7]{huang12}
The joint uninformed capacity is given by
\begin{align}
\capj \sim \frac{1}{2 c^2 \ln 2},
\end{align}
and the optimizing encoder $F$ and collusion channel $\vec{\theta}$ achieving this bound for large $c$ are the arcsine distribution $F^*$ and the interleaving attack $\thint$.
\end{proposition}



\section{Decoders}
\label{sec:dec}

After deriving ``targets'' for our decoders in the previous section, this section discusses decoders that aim to match these bounds. We will follow the score-based framework introduced by Tardos~\cite{tardos03}, which was later generalized to joint decoders by Moulin~\cite{moulin08}. For simple decoding, this means that a user $j$ receives a score $S_j$ of the form
\begin{align}
S_j = \sum_{i=1}^{\ell} S_{j,i} = \sum_{i=1}^{\ell} g(X_{j,i}, y_i, p_i),
\end{align}
where $g$ is called the score function. User $j$ is then accused if $S_j > \eta$ for some threshold $\eta$. 

For joint decoding, scores are assigned to tuples $T = \{j_1, \dots, j_c\}$ of $c$ distinct users according to
\begin{align}
S_T = \sum_{i=1}^{\ell} S_{T,i} = \sum_{i=1}^{\ell} g(X_{j_1,i}, \dots, X_{j_c,i}, y_i, p_i).
\end{align}
In this case, a tuple of users is accused if the joint tuple score exceeds some other threshold $\eta$. Note that this accusation algorithm is not exactly well-defined, since it is possible that a user appears both in a tuple that is accused and in a tuple that is not accused. For the analysis we will assume that the scheme is only successful if the single tuple consisting of all colluders has a score exceeding $\eta$ and no other tuples have a score exceeding $\eta$, in which case all users in the guilty tuple are accused.


\subsection{Simple decoders}
\label{sec:dec-simple}

Several different score functions for the simple decoder setting were considered before, but in this work we will restrict our attention to the following log-likelihood scores, which perform well and turn out to be easy to analyze:
\begin{align}
g(x,y,p) = \ln\left(\frac{\pr_g(x,y|p)}{\pr_i(x,y|p)}\right).
\end{align}
Here $\pr_g(x,y|p) = \pr(X_{j,i} = x, Y_i = y | P_i = p, j \in \mathcal{C})$ corresponds to the probability of seeing the pair $(x,y)$ when user $j$ is guilty, and $\pr_i(x,y|p)$ corresponds to the same probability under the assumption that $j$ is innocent. Using this score function $g$, the complete score of a user is the logarithm of a Neyman-Pearson score over the entire codeword:
\begin{align}
S_j = \sum_{i=1}^{\ell} \ln\left(\frac{\pr_g(x_{j,i},y_i|p_i)}{\pr_i(x_{j,i},y_i|p_i)}\right) 
= \ln\left(\frac{\pr_g(\vec{x}_j, \vec{y}|\vec{p})}{\pr_i(\vec{x}_j, \vec{y}|\vec{p})}\right).
\end{align}
Such Neyman-Pearson scores are known to be optimally discriminative to decide whether to accuse a user or not. Log-likelihood scores were previously considered in the context of fingerprinting in e.g.\ \cite{meerwald12, perez09}.


\subsubsection{Fully informed}
\label{sec:dec-simple-informed}

For the central theorem below, we will make use of the following function $M$, which is closely related to the moment generating functions of scores in one position $i$ for innocent and guilty users. This function is defined as
\begin{align}
M(t) = \sum_{x,y} \pr_i(x,y|p)^{1-t} \pr_g(x,y|p)^{t}
\end{align}
and it satisfies $M(t) = \expn_i e^{t S_{j,i}} = \expn_g e^{(t - 1)S_{j,i}}$
and $M(0) = M(1) = 1$.

\begin{theorem} \label{thm:dec-simple-informed}
Let $p$ and $\vec{\theta}$ be fixed and known to the distributor. Let $\gamma = \ln(1/\eps_2) / \ln(n/\eps_1)$, and let the threshold $\eta$ and code length $\ell$ be defined as
\begin{align}
\eta = \ln\left(\frac{n}{\eps_1}\right), \qquad \ell = \frac{\sqrt{\gamma} (1 + \sqrt{\gamma})}{-\ln M(1 - \sqrt{\gamma})} \, \ln\left(\frac{n}{\eps_1}\right). \label{eq:ellZ}
\end{align}
Then with probability at least $1 - \eps_1$ no innocent users are accused (regardless of which collusion channel was used), and with probability at least $1 - \eps_2$ a colluder is caught (if the collusion channel is indeed $\vec{\theta}$).
\end{theorem}

\begin{proof}
For innocent users $j$, we would like to prove that $\pr_i(S_j > \eta) \leq \eps_1/n$, where $S_j$ is the user's total score over all positions. If this can be proved, then it follows that with probability at least $(1 - \eps_1/n)^n \geq 1 - \eps_1$ no innocent users are accused. Using the Markov inequality for $e^{\alpha S_j}$ with $\alpha > 0$ and optimizing over $\alpha$, we see that the optimum lies close to $\alpha = 1$. For simplicity we choose $\alpha = 1$ which, combined with the given value of $\eta$, leads to the following bound:
\begin{align}
\pr_i(S_j > \eta) = \min_{\alpha > 0} \pr_i(e^{\alpha S_j} > e^{\alpha \eta}) \leq \min_{\alpha > 0} \frac{\expn_i(e^{\alpha S_j})}{e^{\alpha \eta}} \\
= \min_{\alpha > 0} \frac{\prod\limits_{i=1}^{\ell} \expn_i(e^{\alpha S_{j,i}})}{e^{\alpha \eta}} = \min_{\alpha > 0} \frac{M(\alpha)^{\ell}}{(n/\eps_1)^{\alpha}} \leq \frac{M(1)^{\ell}}{n/\eps_1} = \frac{\eps_1}{n}.
\end{align}

For guilty users, we would like to prove that for an arbitrary guilty user $j$, we have $\pr_g(S_j < \eta) \leq \eps_2$. Again using Markov's inequality (but now with a more sophisticated exponent $\beta = \sqrt{\gamma}$) we get
\begin{align}
\hspace{-0.2cm} \pr_g(S_j < \eta) 
 \leq \min_{\beta > 0} \frac{\expn_g(e^{-\beta S_j})}{e^{-\beta \eta}} = \min_{\beta > 0} \frac{\prod\limits_{i=1}^{\ell} \expn_g(e^{-\beta S_{j,i}})}{e^{-\beta \eta}} \\
 = \min_{\beta > 0} \frac{M(1 - \beta)^{\ell}}{e^{-\beta \eta}} \leq \frac{M(1 - \sqrt{\gamma})^{\ell}}{e^{-\sqrt{\gamma} \eta}} = \eps_2,
\end{align}
where the last equality follows from the definitions of $\ell$ and $\eta$ of \eqref{eq:ellZ}. 
\end{proof}

Compared to previous papers analyzing provable bounds on the error probabilities, the proof of Theorem~\ref{thm:dec-simple-informed} is remarkably short and simple. The only problem is that the given expression for $\ell$ is not very informative as to how $\ell$ scales for large $n$. The following corollary answers this question, by showing how $\ell$ scales for small $\gamma$.

\begin{corollary} \label{cor:dec-simple-informed}
If $\gamma = o(1)$ then $\ell$ achieves the optimal asymptotic scaling (achieves capacity) for arbitrary $p$:
\begin{align}
\ell = \frac{\log_2 n}{I(X_1;Y|P = p)}[1 + O(\sqrt{\gamma})], \label{eq:dec-simple}
\end{align}
\end{corollary}

\begin{proof}
First, let us study the behavior of $M(1 - \sqrt{\gamma})$ for small $\gamma$, by computing the first order Taylor expansion of $M(1 - \sqrt{\gamma})$ around $\gamma = 0$:
\begin{align}
& M(1 - \sqrt{\gamma}) \nonumber \\
&= \sum_{x,y} \pr_g(x,y|p) \exp\left(-\sqrt{\gamma} \ln\left(\frac{\pr_g(x,y|p)}{\pr_i(x,y|p)}\right)\right) \\
 &\stackrel{(a)}{=} \sum_{x,y} \pr_g(x,y|p) \left(1 - \sqrt{\gamma} \ln\left(\frac{\pr_g(x,y|p)}{\pr_i(x,y|p)}\right) + O(\gamma)\right) \\
 &= 1 - \sqrt{\gamma} \sum_{x,y} \pr_g(x,y|p) \ln\left(\frac{\pr_g(x,y|p)}{\pr_i(x,y|p)}\right) + O(\gamma) \\
 &= 1 - \sqrt{\gamma} I(X_1;Y|P = p) \ln 2 + O(\gamma).
\end{align}
Here $(a)$ follows from the fact that if $\pr_g(x,y|p) = 0$, the factor $\pr_g(x,y|p)$ in front of the exponentiation would already cause this term to be $0$, while if $\pr_g(x,y|p) > 0$, then also $\pr_i(x,y|p) > 0$ and thus the ratio is bounded and does not depend on $\gamma$. Substituting the above result in the original equation for $\ell$ we thus get the result of \eqref{eq:dec-simple}:
\begin{align}
\ell &= \frac{\sqrt{\gamma} (1 + \sqrt{\gamma})}{-\ln M(1 - \sqrt{\gamma})} \ln\left(\frac{n}{\eps_1}\right) \\
 &= \frac{\sqrt{\gamma} (1 + \sqrt{\gamma})}{\sqrt{\gamma} I(X_1;Y|P = p) \ln 2 + O(\gamma)} \ln\left(\frac{n}{\eps_1}\right) \\
 &= \frac{\log_2 n}{I(X_1;Y|P = p)}[1 + O(\sqrt{\gamma})].
\end{align}
Since the capacities tell us that $\ell / \log_2(n) \gtrsim I(X_1;Y|P = p)^{-1}$, it follows that $\ell$ asymptotically achieves capacity.
\end{proof}


Since this construction is asymptotically optimal regardless of $p$, in the fully informed setting we can now simply optimize $p$ (using Theorem~\ref{thm:cap-simple-informed}) to get the following results.

\begin{corollary}
Using the values for $p$ of Theorem~\ref{thm:cap-simple-informed}, the asymptotics for $\ell$ for the five attacks of Section~\ref{sec:attacks} are:
\begin{align}
\ell(\thint) &= 2 c^2 \ln(n) \left[1 + O(\sqrt{\gamma})\right], \\
\ell(\thall) &= \frac{c}{\ln(2)^2} \ln(n) \left[1 + O(\sqrt{\gamma})\right], \\
\ell(\thmaj) &= \pi c \ln(n) \left[1 + O(\sqrt{\gamma})\right], \\
\ell(\thmin) &= \frac{c}{\ln(2)^2} \ln(n) \left[1 + O(\sqrt{\gamma})\right], \\
\ell(\thcoi) &= \frac{4c}{\ln(2)^2} \ln(n) \left[1 + O(\sqrt{\gamma})\right].
\end{align}
\end{corollary}

Since the all-$1$ attack is equivalent to group testing, we mention this result separately, together with a more explicit expression for $g$. 

\begin{corollary} \label{cor:dec-simple-informed-group}
Let $\vec{\theta} = \thall$ and let $p \approx \ln(2)/c$ be fixed. Then the log-likelihood score function $g$ is given by\footnote{To be precise: $g(0,1) = c \log_2(2 - 2^{1/c})$, and for convenience we have scaled $g$ by a factor $c \ln 2$.}
\begin{align}
g(x,y) = \begin{cases}
+1 & (x,y) = (0,0) \\
-1 + O(1/c) & (x,y) = (0,1) \\
-\infty & (x,y) = (1,0) \\
+c & (x,y) = (1,1)
\end{cases}
\end{align}
\sloppypar{Using this score function in combination with the parameters $\eta$ and $\ell$ of Theorem~\ref{thm:dec-simple-informed}, we obtain a simple group testing algorithm with an asymptotic number of group tests of 
\begin{align}
\ell \sim \frac{c \ln n}{\ln(2)^2} \approx 2.08 c \ln n,
\end{align}
thus achieving the simple group testing capacity.}
\end{corollary}


\subsubsection{Partially informed}
\label{sec:dec-simple-part}

Since the score functions from Section~\ref{sec:dec-simple-informed} achieve capacity for each value of $p$, using this score function we also trivially achieve the partially informed capacity when the arcsine distribution is used. Estimates of these capacities, and thus the resulting code lengths, can be found in Section~\ref{sec:cap-simple-part}.


\subsubsection{Uninformed}
\label{sec:dec-simple-uninformed}

We now arrive at what is arguably one of the most important results of this paper. Just like Oosterwijk et al.~\cite{oosterwijk13}, who specifically studied the score function $h$ tailored against the interleaving attack, we now also take a closer look at the log-likelihood score function designed against the interleaving attack. \footnote{Considering the interleaving attack for designing a universal decoder is further motivated by the results of Abbe and Zheng~\cite{abbe10, meerwald12}, who showed that under certain conditions, the worst-case attack decoder is a universal capacity-achieving decoder. The interleaving attack is theoretically not the worst-case attack for finite $c$, but since it is known to be the asymptotic worst-case attack, the difference between the worst-case attack and the interleaving attack vanishes for large $c$.} Working out the details, this score function is of the form:
\begin{align}
g(x,y,p) = \begin{cases} 
\ln\left(1 + \frac{p}{c(1-p)}\right) & x = y = 0 \\
\ln\left(1 - \frac{1}{c}\right) & x \neq y \\
\ln\left(1 + \frac{1 - p}{cp}\right) & x = y = 1
\end{cases} \label{eq:simple}
\end{align}
The first thing to note here is that if we denote Oosterwijk et al.'s~\cite{oosterwijk13b} score function by $h$, then $g$ satisfies
\begin{align}
g(x,y,p) = \ln\left(1 + \frac{h(x,y,p)}{c}\right).
\end{align}
If $h(x,y,p) = o(c)$, then by Tayloring the logarithm around $c = \infty$ we see that $g \approx h/c$. Since scaling a score function by a constant does not affect its performance, this implies that $g$ and $h$ are then equivalent. Since for Oosterwijk et al.'s score function one generally needs to use \textit{cut-offs} on $F$ that guarantee that $h(x,y,p) = o(c)$ (cf.~\cite{ibrahimi13}), and since the decoder of Oosterwijk et al.\ is known to asymptotically achieve the uninformed capacity, we immediately get the following result.

\begin{proposition} \label{prop:dec-simple-uninformed}
The score function $g$ of \eqref{eq:simple} asymptotically achieves the uninformed simple capacity when the same cut-offs on $F$ as those in~\cite{ibrahimi13} are used.
\end{proposition}

So optimizing the decoder so that it is resistant against the interleaving attack again leads to a decoder that is resistant against arbitrary attacks.

\paragraph{Cutting off the cut-offs}

Although Proposition~\ref{prop:dec-simple-uninformed} is already a nice result, we can do even better. We can prove a stronger statement, which shows one of the reasons why the log-likelihood decoder is probably more practical than the decoder of Oosterwijk et al.

\begin{theorem}
The score function $g$ of \eqref{eq:simple} achieves the uninformed simple capacity when \underline{\textbf{no cut-offs}} are used.
\end{theorem}

\begin{proof}[sketch]
First note that in the limit of large $c$, the cut-offs of Ibrahimi et al.\ converge to $0$. So for large $c$, the difference between not using cut-offs and using cut-offs is negligible, as long as the contribution of the tails of $p$ near $0$ or $1$ to the distribution of user scores is negligible. Since with this score function $g$, all moments of both innocent and guilty user scores are finite (arbitrary powers of logarithms always lose against the $1/\sqrt{p(1-p)}$ of the arcsine distribution and the decreasing width of the interval between $0$ and the cut-off), the tails indeed decay exponentially. So also without cut-offs this score function asymptotically achieves the uninformed simple capacity.
\end{proof}

Note that the same result does not apply to the score function of Oosterwijk et al.~\cite{oosterwijk13b}, for which the tails of the distributions are not Gaussian enough to omit the use of cut-offs. The main difference is that for small $p$, the score function $h$ of \cite{oosterwijk13} scales as $1/p$ (which explodes when $p$ is really small), while the log-likelihood decoder $g$ then only scales as $\ln(1/p)$ which is much smaller.

\paragraph{All roads lead to Rome}

Let us now mention a third way to obtain a capacity-achieving uninformed simple decoder which is again very similar to the two decoders above. To construct this decoder, we use a Bayesian approximation of the proposed empirical mutual information decoder of Moulin~\cite{moulin08}, and again plug in the asymptotic worst-case attack, the interleaving attack.

\begin{theorem}
Using Bayesian inference with an a priori probability of guilt of $\pr(j \in \mathcal{C}) = \frac{c}{n}$, the empirical mutual information decoder tailored against the interleaving attack can be approximated with the following score function:
\begin{align}
m(x,y,p) = \begin{cases}
\ln\left(1 + \frac{p}{n(1 - p)}\right) & x = y = 0 \\
\ln\left(1 - \frac{1}{n}\right) & x \neq y \\ 
\ln\left(1 + \frac{1 - p}{np}\right) & x = y = 1
\end{cases} \label{eq:simple2}
\end{align}
\end{theorem}

\begin{proof}[sketch]
For now, let $p$ be fixed. The empirical mutual information decoder assigns a score $S_j$ to a user $j$ using
\begin{align}
S_j = \sum_{x,y} \hat{\pr}(x,y) \ln\left(\frac{\hat{\pr}(x,y)}{\hat{\pr}(x)\hat{\pr}(y)}\right) = \sum_{i=1}^{\ell} \ln\left(\frac{\hat{\pr}(x_{j,i},y_i)}{\hat{\pr}(x_{j,i})\hat{\pr}(y_i)}\right)
\end{align}
where $\hat{\pr}(\cdot)$ denotes the empirical estimate of $\pr(\cdot)$ based on the data $\vec{x}$, $\vec{y}$, $\vec{p}$. For large $\ell$, these estimates will converge to the real probabilities, so we can approximate $S_j$ by
\begin{align}
S_j \approx \sum_{i=1}^{\ell} \ln\left(\frac{\pr(x_{j,i},y_i)}{\pr(x_{j,i})\pr(y_i)}\right) = \sum_{i=1}^{\ell} m(x_{j,i},y_i,p_i).
\end{align}
Here $\pr(X_{j,i})$ and $\pr(y_i)$ can be easily computed, but for computing $\pr(X_{j,i}, y_i)$ we need to know whether user $j$ is guilty or not. Using Bayesian inference, we can write
\begin{align}
\pr(x,y) = \pr_g(x,y) \pr(j \in \mathcal{C}) + \pr_i(x,y) \pr(j \notin \mathcal{C}).
\end{align}
Assuming an a priori probability of guilt of $\pr(j \in \mathcal{C}) = c/n$, we can work out the details to obtain
\begin{align}
m(x,y,p) = \ln\left(1 + \frac{c}{n}\left[\frac{\pr_g(x,y)}{\pr_i(x,y)} - 1\right]\right).
\end{align}
Filling in the corresponding probabilities for the interleaving attack, we end up with the score function of \eqref{eq:simple2}.
\end{proof}

For values of $p$ with $o(1) < p < 1 - o(1)$, this decoder is again equivalent to both the log-likelihood score function $g$ and Oosterwijk et al.'s score function $h$. 


\subsection{Joint decoders}
\label{sec:dec-joint}

For the joint decoding setting, scores are assigned to tuples of $c$ users, and again higher scores correspond to a higher probability of being accused. The most natural step from the simple log-likelihood decoders to joint decoders seems to be to use the following joint score function:
\begin{align}
g(x_1, \dots, x_c,y,p) = \ln\left(\frac{\pr_{g^c}(x_1, \dots, x_c,y|p)}{\pr_{i^c}(x_1, \dots, x_c,y|p)}\right).
\end{align}
Here $\pr_{g^c}(\cdot)$ is under the assumption that in this tuple \textit{all users are guilty}, while for $\pr_{i^c}(\cdot)$ we assume that \textit{all users are innocent}. Note that under the assumption that the attack is colluder-symmetric, the score function only depends on $z = \sum_{i=1}^c x_i$:
\begin{align}
g(x_1, \dots, x_c,y,p) = g(z,y,p) = \ln\left(\frac{\pr_{g^c}(z,y|p)}{\pr_{i^c}(z,y|p)}\right).
\end{align}


\subsubsection{Fully informed}
\label{sec:dec-joint-informed}

To analyze the joint decoder, we again make use of the moment generating function for the score assigned to tuples of $c$ innocent users. This function is now defined by
\begin{align}
M(t) = \sum_{z,y} \pr_{i^c}(z,y|p)^{1-t} \pr_{g^c}(z,y|p)^{t}
\end{align}
and it satisfies $M(t) = \expn_{i^c} e^{t S_{j,i}} = \expn_{g^c} e^{(t - 1)S_{j,i}}$ and $M(0) = M(1) = 1$. Using similar techniques as in Section~\ref{sec:dec-simple-informed}, we obtain the following result.

\begin{theorem} \label{thm:dec-joint-informed}
Let $p$ and $\vec{\theta}$ be fixed and known to the distributor. Let $\gamma = \ln(1/\eps_2) / \ln(n^c/\eps_1)$, and let the threshold $\eta$ and code length $\ell$ be defined as
\begin{align}
\eta = \ln\left(\frac{n^c}{\eps_1}\right), \qquad \ell = \frac{\sqrt{\gamma} (1 + \sqrt{\gamma})}{-\ln M(1 - \sqrt{\gamma})} \ln\left(\frac{n^c}{\eps_1}\right).
\end{align}
Then with probability at least $1 - \eps_1$ all all-innocent tuples are not accused, and with probability at least $1 - \eps_2$ the single all-guilty tuple is accused.
\end{theorem}

\begin{proof}[sketch]
The proof is very similar to the proof of Theorem~\ref{thm:dec-simple-informed}. Instead of $n$ innocent and $c$ guilty users we now have $\binom{n}{c} < n^c$ all-innocent tuples and just $1$ all-guilty tuple, which changes some of the numbers in $\gamma$, $\eta$ and $\ell$. We again apply the Markov inequality with $\alpha = 1$ for innocent tuples and $\beta = \sqrt{\gamma}$ for guilty tuples, to obtain the given expressions for $\eta$ and $\ell$.
\end{proof}

Note that Theorem~\ref{thm:dec-joint-informed} does not prove that we can actually find the set of colluders with high probability, since mixed tuples consisting of both innocent and guilty users also exist, and these may or may not have a score exceeding $\eta$. This does prove that with high probability we can find a set $\mathcal{C}'$ of $c$ users, for which (i) all tuples not containing these users have a score below $\eta$, and (ii) the tuple containing exactly these users has a score above $\eta$. Regardless of what the scores for mixed tuples are, with probability at least $1 - \eps_1 - \eps_2$ such a set consists and contains at least one colluder. Furthermore, if this set $\mathcal{C}'$ is unique, then with high probability this is exactly the set of colluders. But there is no guarantee that it is unique without additional proofs. This is left for future work.

To further motivate why using this joint decoder may be the right choice, the following proposition shows that at least the scaling of the resulting code lengths is optimal. Note that the extra $c$ that we get from $\ln(n^c) = c \ln n$ can be combined with the mutual information $I(\cdot)$ to obtain $\frac{1}{c} I(\cdot)$, which corresponds to the joint capacity.

\begin{proposition} \label{prop:dec-joint-informed}
If $\gamma = o(1)$ then the code length $\ell$ of Theorem~\ref{thm:dec-joint-informed} scales as
\begin{align}
\ell = \frac{\log_2 n}{\frac{1}{c}I(Z;Y|P = p)}\left[1 + O(\sqrt{\gamma})\right],
\end{align}
thus asymptotically achieving the optimal code length (up to first order terms) for arbitrary values of $p$.
\end{proposition}

Since the asymptotic code length is optimal regardless of $p$, these asymptotics are also optimal when $p$ is optimized to maximize the mutual information in the fully informed setting.

Finally, although it is hard to estimate the scores of mixed tuples with this decoder, just like in~\cite{oosterwijk14} we expect that the joint decoder score for a tuple is roughly equal to the sum of the $c$ individual simple decoder scores. So a tuple of $c$ users consisting of $k$ colluders and $c - k$ innocent users is expected to have a score roughly a factor $k/c$ smaller than the expected score for the all-guilty tuple. So after computing the scores for all tuples of size $c$, we can get rough estimates of how many guilty users are contained in each tuple, and for instance try to find the set $\mathcal{C}'$ of $c$ users that best matches these estimates. There are several options for post-processing that may improve the accuracy of using this joint decoder, which are left for future work.


\subsubsection{Partially informed}
\label{sec:dec-joint-part}

As mentioned in Proposition~\ref{prop:dec-joint-informed}, the code length is asymptotically optimal regardless of $p$, so the code length in the partially uninformed setting is also asymptotically optimal. Asymptotics on $\ell$ can thus be obtained by combining Proposition~\ref{prop:dec-joint-informed} with the results of Section~\ref{sec:cap-joint-part}. 


\subsubsection{Uninformed}
\label{sec:dec-joint-uninformed}

Note that if the above joint decoder turns out to work well, then we can again plug in the interleaving attack to get something that might just work well against arbitrary attacks. While we cannot prove that this joint decoder is optimal, we can already see what the score function would be, and conjecture that it works against arbitrary attacks.

\begin{conjecture}
The joint log-likelihood decoder against the interleaving attack, with the score function $g$ defined by
\begin{align}
g(z, y, p) = \begin{cases}
\ln(1 - \frac{z}{c}) - \ln(1 - p) \quad & (y = 0) \\
\ln(\frac{z}{c}) - \ln(p) & (y = 1)
\end{cases}
\end{align}
works against arbitrary attacks and asymptotically achieves the joint capacity of the uninformed fingerprinting game.
\end{conjecture}

A further study of this universal joint decoder is left as an open problem.


\section{Discussion}
\label{sec:discussion}

Let us now briefly discuss the results from Sections~\ref{sec:cap} and \ref{sec:dec}, their consequences, and some directions for future work.

\paragraph{Informed simple decoding}

For the setting of simple decoders, we derived explicit asymptotics on the informed capacities for various attacks, which often scale as $\Theta(c^{-1})$. We further showed that log-likelihood scores provably match these bounds for large $n$, regardless of $\vec{\theta}$ and $p$. Because these decoders are optimal for any value of $p$, they are also optimal in the partially informed setting, where different values of $p$ are used. If the encoder uses the arcsine distribution to generate biases, we showed that these capacities generally seem to scale as $\Theta(c^{-3/2})$, which is roughly `halfway' between the fully informed and uninformed capacities.

\paragraph{Uninformed simple decoding}

Although log-likelihood decoders have already been studied before in the context of fingerprinting, the main drawback was always that to use these decoders, you would either have to fill in (and know) the exact pirate strategy, or compute the worst-case attack explicitly. So if you are in the simple uninformed setting where you don't know the pirate strategy and where the worst-case attack is not given by a nice closed-form expression~\cite[Fig.~4b]{huang12}, how can you construct such decoders for large $c$? The trick seems to be to just fill in the \textit{asymptotic} worst-case attack, which Huang and Moulin showed is the interleaving attack~\cite{huang12}, and which is much simpler to analyze. After previously suggesting this idea to Oosterwijk et al., we now used the same trick here to obtain two other capacity-achieving score functions using two different methods (but each time filling in the interleaving attack). So in total we now have three different methods to obtain (closed-form) capacity-achieving decoders in the uninformed setting:
\begin{itemize}
  \item Using Lagrange-multipliers, Oosterwijk et al.~\cite{oosterwijk13} obtained:
  \begin{align}
h(x,y,p) = \begin{cases} 
+\frac{p}{1 - p} \qquad & x = y = 0 \\
-1 & x \neq y \\
+\frac{1 - p}{p} & x = y = 1
\end{cases}
  \end{align}
  \item Using Neyman-Pearson-based log-likelihood scores, we obtained:
  \begin{align}
g(x,y,p) = \begin{cases} 
\ln\left(1 + \frac{p}{c(1-p)}\right) & x = y = 0 \\
\ln\left(1 - \frac{1}{c}\right) & x \neq y \\
\ln\left(1 + \frac{1 - p}{cp}\right) & x = y = 1
\end{cases}
  \end{align}
  \item Using a Bayesian approximation of the empirical mutual information decoder of Moulin~\cite{moulin08}, we obtained:
\begin{align}
m(x,y,p) = \begin{cases}
\ln\left(1 + \frac{p}{n(1 - p)}\right) & x = y = 0 \\
\ln\left(1 - \frac{1}{n}\right) & x \neq y \\ 
\ln\left(1 + \frac{1 - p}{np}\right) & x = y = 1
\end{cases}
\end{align}
\end{itemize}
For $o(1) < p < 1 - o(1)$ and large $c, n$, these score functions are equivalent up to a scaling factor:
\begin{align}
h(x,y,p) \sim c \cdot g(x,y,p) \sim n \cdot m(x,y,p),
\end{align}
and therefore all three are asymptotically optimal. So there may be many different roads that lead to Rome, but they all seem to have one thing in common: to build a universal decoder that works against arbitrary attacks, one should build a decoder that works against the asymptotic worst-case pirate attack, the interleaving attack. And if it does work against this attack, then it probably works against any other attack as well.

\paragraph{Joint decoding}

Although deriving the joint informed capacities is much easier than deriving the simple informed capacities, actually building decoders that provably match these bounds is a different matter. We conjectured that the same log-likelihood scores achieve capacity when a suitable accusation algorithm is used, and we conjectured that the log-likelihood score built against the interleaving attack achieves the uninformed joint capacity, but we cannot prove any of these statements beyond reasonable doubt. For now this is left as an open problem.

\paragraph{Group testing}

Since the all-$1$ attack is equivalent to group testing, some of the results we obtained also apply to group testing. The joint capacity was already known~\cite{sebo85}, but to the best of our knowledge both the simple capacity~(Corollary~\ref{cor:cap-simple-informed-group}) and a simple decoder matching this simple capacity~(Corollary~\ref{cor:dec-simple-informed-group}) were not yet known before. Attempts have been made to build efficient simple decoders with a code length not much longer than the joint capacity~\cite{chan12}, but these do not match the simple capacity. Future work will include computing the capacities and building decoders for various noisy group testing models, where the marking assumption may not apply. 

\paragraph{Dynamic fingerprinting}

Although this paper focused on applications to the `static' fingerprinting game, the construction of~\cite{laarhoven13tit} can trivially be applied to the decoders in this paper as well to build efficient dynamic fingerprinting schemes. Although the asymptotics for the code length in this dynamic construction are the same, (i) the order terms are significantly smaller in the dynamic game, (ii) one does not need the assumption that the pirate strategy is colluder-symmetric, and (iii) one does not necessarily need to know (a good estimate of) $c$ in advance~\cite[Section~V]{laarhoven13tit}. An important open problem remains to determine the dynamic uninformed fingerprinting capacity, which may prove or disprove that the construction of \cite{laarhoven13tit} is optimal.

\paragraph{Further generalizations}

While this paper already aims to provide a rather complete set of guidelines on what to do in the various different fingerprinting games (with different amounts of side-information, and different computational assumptions on the decoder), there are some further generalizations that were not considered here due to lack of space. We mention two in particular:
\begin{itemize}
  \item \textbf{Larger alphabets}: In this work we focused on the binary case of $q = 2$ different symbols, but it may be advantageous to work with larger alphabet sizes $q > 2$, since the code length decreases linearly with $q$. For the results about decoders we did not really use that we were working with a binary alphabet, so it seems a straightforward exercise to prove that the $q$-ary versions of the log-likelihood decoders also achieve capacity. A harder problem seems to be to actually compute these capacities in the various informed settings, since the maximization problem then transforms from a one-dimensional optimization problem to a $(q - 1)$-dimensional optimization problem.
  \item \textbf{Tuple decoding}: As in~\cite{oosterwijk14}, we can consider a setting in between the simple and joint decoding settings, where decisions to accuse are made based on looking at tuples of users of size at most $t$. Tuple decoding may offer a trade-off between the high complexity, low code length of a joint decoder and the low complexity, higher code length of a simple decoder, and so it may be useful to know how the capacities scale in the region $1 < t < c$.
\end{itemize}


\section{Acknowledgments} 
The author is very grateful to Pierre Moulin for his insightful comments and suggestions during the author's visit to Urbana-Champaign that inspired work on this paper. The author would also like to thank Teddy Furon for pointing out the connection between decoders designed against the interleaving attack and the results of Abbe and Zheng~\cite{abbe10}, and for finding some mistakes in a preliminary version of this manuscript. Finally, the author thanks Jeroen Doumen, Jan-Jaap Oosterwijk, Boris \v{S}kori\'{c}, and Benne de Weger for valuable discussions and comments.


\end{document}